\documentclass[10pt]{article}
\usepackage{geometry}                
\usepackage{graphicx}
\usepackage{amssymb}
\usepackage{epstopdf}
\usepackage{url}
\usepackage{authblk}

\usepackage{epsfig,bm,epsf,float, graphicx, amssymb, verbatim, amsmath, amsthm, algorithm, algorithmic, mathrsfs}
\usepackage{subfig}
\usepackage{pstricks,pst-node}
\usepackage{pst-plot}
\usepackage{cite}
\usepackage{balance}

\usepackage{amsthm}
\usepackage{wrapfig}
\usepackage{algorithm}
\usepackage{algorithmic}
\usepackage{lpic}

\def\all{all}
\ifx\files\all
\fi

\usepackage{color}

\title{Comment on ``Asymptotic Achievability of the Cram\'{e}r-Rao Bound for
Noisy Compressive Sampling''}

\author{Behtash Babadi\thanks{Department of Electrical and Computer Engineering, University of Maryland, College Park, MD 20742 USA (e-mail: behtash@umd.edu).}, Nicholas Kalouptsidis\thanks{Department of Informatics and Telecommunications, National and Kapodistrian University of Athens, Athens 157 73, Greece (e-mail: kalou@di.uoa.gr).}, and Vahid Tarokh\thanks{John A. Paulson School of Engineering and Applied Sciences,
Harvard University, Cambridge, MA 02138 USA (e-mail: vahid@seas.harvard.edu).}}

\begin{document}
\maketitle

%


\newcommand{\Mp}{\bm \phi}
\newcommand{\supp}{\operatorname{{\mathrm supp}}}


\newtheorem{thm}{Theorem}
\newtheorem{cor}[thm]{Corollary}
\newtheorem{prop}[thm]{Proposition}
\newtheorem{lem}[thm]{Lemma}
\newtheorem{defn}[thm]{Definition}

\newcommand{\sgn}{\operatorname{{\mathrm sgn}}}
\newcommand{\var}{\operatorname{{\mathrm var}}}
\newcommand{\diag}{\operatorname{{\mathrm diag}}}

\def\QED{{\setlength{\fboxsep}{0pt}\setlength{\fboxrule}{0.2pt}\fbox{\rule[0pt]{0pt}{1.3ex}\rule[0pt]{1.3ex}{0pt}}}} 


\def\cal#1{\mathcal{#1}}

\def\avg#1{\left< #1 \right>}
\def\abs#1{\left| #1 \right|}
\def\recip#1{\frac{1}{#1}}
\def\vhat#1{\hat{{\bf #1}}}
\def\smallfrac#1#2{{\textstyle\frac{#1}{#2}}}
\def\smallrecip#1{\smallfrac{1}{#1}}

\def\spshalf{{1\over{2}}}
\def\Orabi{\Omega_{\rm rabi}}
\def\btt#1{{\tt$\backslash$#1}}


\def\schrod{Schroedinger's Equation}
\def\helm{Helmholtz Equation}

\def\be{\begin{equation}}
\def\ee{\end{equation}}
\def\bea{\begin{eqnarray}}
\def\eea{\end{eqnarray}}
\def\bean{\begin{mathletters}\begin{eqnarray}}
\def\eean{\end{eqnarray}\end{mathletters}}

\newcommand{\tbox}[1]{\mbox{\tiny #1}}
\newcommand{\half}{\mbox{\small $\frac{1}{2}$}}
\newcommand{\pit}{\mbox{\small $\frac{\pi}{2}$}}
\newcommand{\sfrac}[1]{\mbox{\small $\frac{1}{#1}$}}
\newcommand{\mbf}[1]{{\mathbf #1}}
\def\text{\tbox}

\newcommand{\mV}{{\mathsf{V}}}
\newcommand{\mL}{{\mathsf{L}}}
\newcommand{\mA}{{\mathsf{A}}}
\newcommand{\lB}{\lambda_{\tbox{B}}}  
\newcommand{\ofr}{{(\mbf{r})}}       
\def\ofkr{(k;\mbf{r})}			
\def\ofks{(k;\mbf{s})}			
\newcommand{\ofs}{{(\mbf{s})}}       
\def\xt{\mbf{x}^{\tbox T}}		

\def\ce{\tilde{C}_{\tbox E}}		
\def\cew{\tilde{C}_{\tbox E}(\omega)}		
\def\ceqmw{\tilde{C}^{\tbox{qm}}_{\tbox E}(\omega)}	
\def\cewqm{\tilde{C}^{\tbox{qm}}_{\tbox E}}	
\def\ceqm{C^{\tbox{qm}}_{\tbox E}}	
\def\cw{\tilde{C}(\omega)}		
\def\cfw{\tilde{C}_{\cal F}(\omega)}		

\def\tcl{\tau_{\tbox{cl}}}		
\def\tcol{\tau_{\tbox{col}}}		
\def\terg{t_{\tbox{erg}}}		
\def\tbl{\tau_{\tbox{bl}}}		
\def\theis{t_{\tbox{H}}}		

\def\area{\mathsf{A}_D}			
\def\ve{\nu_{\tbox{E}}}			
\def\vewna{\nu_E^{\tbox{WNA}}}		

\def\dxcqm{\delta x^{\tbox{qm}}_{\tbox c}}	

\newcommand{\rop}{\hat{\mbf{r}}}	
\newcommand{\pop}{\hat{\mbf{p}}}

\newcommand{\sint}{\oint \! d\mbf{s} \,} 
\def\gint{\oint_\Gamma \!\! d\mbf{s} \,} 
\newcommand{\lint}{\oint \! ds \,}	
\def\infint{\int_{-\infty}^{\infty} \!\!}	
\def\dn{\partial_n}				
\def\aswapb{a^*\!{\leftrightarrow}b}		
\def\eps{\varepsilon}				

\def\dhdxt{\partial {\cal H} / \partial x}
\def\dhdx{\pd{\cal H}{x}}
\def\dhdxnm{\left( \pd{\cal H}{x} \right)_{\!nm}}
\def\dhdxnmsq{\left| \left( \pd{\cal H}{x} \right)_{\!nm} \right| ^2}

\def\bcs{\stackrel{\tbox{BCs}}{\longrightarrow}}	

\def\wx{\omega_x}
\def\wy{\omega_y}
\newcommand{\ofro}{({\bf r_0})}
\def\Eb{E_{\rm blue,rms}}
\def\Er{E_{\rm red,rms}}
\def\Es2{E_{0,{\rm sat}}^2}
\def\sb{s_{\rm blue}}
\def\sr{s_{\rm red}}

\def\ie{{\it i.e.\ }}
\def\eg{{\it e.g.\ }}
\newcommand{\etal}{{\it et al.\ }}
\newcommand{\ibid}{{\it ibid.\ }}

\def\gap{\hspace{0.2in}}

%

\newcounter{eqletter}
\def\mathletters{%
\setcounter{eqletter}{0}%
\addtocounter{equation}{1}
\edef\curreqno{\arabic{equation}}
\edef\@currentlabel{\theequation}
\def\theequation{%
\addtocounter{eqletter}{1}\thechapter.\curreqno\alph{eqletter}%
}%
}
\def\endmathletters{\setcounter{equation}{\curreqno}}


\newcommand{\bk}{{\bf k}}
\def\kf{k_{\text F}}
\newcommand{\br}{{\bf r}}
\newcommand{\TL}{{\text{(L)}}}
\newcommand{\TR}{{\text{(R)}}}
\newcommand{\TLR}{{\text{L,R}}}
\newcommand{\VSD}{V_{\text{SD}}}
\newcommand{\GT}{\Gamma_{\text{T}}}
\newcommand{\DEL}{\mbox{\boldmath $\nabla$}}
\def\lf{\lambda_{\text F}}
\def\st{\sigma_{\text T}}
\def\stlr{\sigma_{\text T}^{\text{L$\rightarrow$R}}}
\def\strl{\sigma_{\text T}^{\text{R$\rightarrow$L}}}
\def\aeff{a_{\text{eff}}}
\def\aaeff{A_{\text{eff}}}
\def\gat{G_{\text{atom}}}
\newcommand{\LB}{Landauer-B\"{u}ttiker}

%
%

In \cite{babadi}, we proved the asymptotic achievability of the Cram\'{e}r-Rao bound in the compressive sensing setting in the linear sparsity regime. In the proof, we used an erroneous closed-form expression of $\alpha \sigma^2$ for the genie-aided Cram\'{e}r-Rao bound $\sigma^2 \textrm{Tr} (\mathbf{A}^*_\mathcal{I} \mathbf{A}_\mathcal{I})^{-1}$ from Lemma 3.5, which appears in Eqs. (20) and (29). The proof, however, holds if one avoids replacing $\sigma^2 \textrm{Tr} (\mathbf{A}^*_\mathcal{I} \mathbf{A}_\mathcal{I})^{-1}$ by the expression of Lemma 3.5, and hence the claim of the Main Theorem stands true. 

In Chapter 2 of the Ph. D. dissertation by Behtash Babadi \cite{phd}, this error was fixed and a more detailed proof in the non-asymptotic regime was presented. A draft of Chapter 2 of \cite{phd} is included in this note, verbatim. We would like to refer the interested reader to the full dissertation, which is electronically archived in the ProQuest database \cite{phd}, and a draft of which can be accessed through the author's homepage under: \url{http://ece.umd.edu/~behtash/babadi_thesis_2011.pdf}.


%
%
%
%

\section*{Asymptotic Achievability of the Cram\'{e}r-Rao Bound for Noisy Compressive
Sampling\footnote{Chapter 2 of \cite{phd}.}} \label{ch2}

\section{Introduction}
In this chapter, we consider the problem of estimating a sparse vector based on noisy observations. Suppose that we have a compressive sampling (Please see \cite{candes1} and \cite{donoho}) model of the form
\begin{equation}\label{model}
\mathbf{y} = \mathbf{A} \mathbf{x} + \mathbf{n}
\end{equation}
where $\mathbf{x} \in \mathbb{C}^M$ is the unknown sparse vector to be estimated, $\mathbf{y} \in \mathbb{C}^N$ is the observation vector, $\mathbf{n} \sim \mathcal{N}(0,\sigma^2 \mathbf{I}_N) \in \mathbb{C}^N$ is the Gaussian noise and $\mathbf{A} \in \mathbb{C}^{N \times M}$ is the measurement matrix. Suppose that $\mathbf{x}$ is sparse, \emph{i.e.}, $\| \mathbf{x} \|_0 = L < M$. Let $\mathcal{I} := {\sf supp}(\mathbf{x})$ and
\begin{eqnarray}
\nonumber \alpha &:=& L/N\\
\nonumber \beta &:=& M/L > 2
\end{eqnarray}
be fixed numbers.

The estimator must both estimate the locations and the values of the non-zero elements of $\mathbf{x}$. If a Genie provides us with $\mathcal{I}$, the problem reduces to estimating the values of the non-zero elements of $\mathbf{x}$. We denote the estimator to this reduced problem by Genie-Aided estimator (GAE).

Clearly, the mean squared estimation error (MSE) of any unbiased estimator is no less than that of the GAE (see \cite{candes2}), since the GAE does not need to estimate the locations of the nonzero elements of $\mathbf{x}$ $\big($ $\log_2 {M \choose L} \doteq M H(1/\beta)$ bits, where $H(\cdot)$ is the binary entropy function$\big)$.

Recently, Haupt and Nowak \cite{nowak} and Cand\`{e}s and Tao \cite{candes2} have proposed estimators which achieve the estimation error of the GAE up to a factor of $\log M$. In \cite{nowak}, a measurement matrix based on Rademacher projections is constructed and an iterative bound-optimization recovery procedure is proposed. Each step of the procedure requires $O(MN)$ operations and the iterations are repeated until convergence is achieved. It has been shown that the estimator achieves the estimation error of the GAE up to a factor of $\log M$.

Cand\`{e}s and Tao \cite{candes2} have proposed an estimator based on linear programming, namely the Dantzig Selector, which achieves the estimation error of the GAE up to a factor of $\log M$, for Gaussian measurement matrices. The Dantzig Selector can be recast as a linear program and can be efficiently solved by the well-known primal-dual interior point methods, as suggested in \cite{candes2}. Each iteration requires solving an $M \times M$ system of linear equations and the iterations are repeated until convergence is attained.

We construct an estimator based on Shannon theory and the notion of typicality \cite{cover} that asymptotically achieves the Cram\'{e}r-Rao bound on the estimation error of the GAE without the knowledge of the locations of the nonzero elements of $\mathbf{x}$, for Gaussian measurement matrices. Although the estimator presented in this chapter has higher complexity (exponential) compared to the estimators in \cite{candes2} and \cite{nowak}, to the best of our knowledge it is the first result establishing the achievability of the Cram\'{e}r-Rao bound for noisy compressive sampling \cite{babadi}. The problem of finding efficient and low-complexity estimators that achieve the Cram\'{e}-Rao bound for noisy compressive sampling still remains open.

The outline of this chapter follows next. In Section \ref{main_result}, we state the main result of this chapter and present its proof in Section \ref{proof}.

\section{Main Result}\label{main_result}

The main result of this section is the following:
\textit{\begin{thm}[Main Theorem]\label{main_thm_crb}
In the compressive sampling model of $\mathbf{y} = \mathbf{A} \mathbf{x} + \mathbf{n}$, let $\textbf{A}$ be a measurement matrix whose elements are i.i.d. Gaussian $\mathcal{N}(0,1)$. Let $e_G^{(M)}$ be the Cram\'{e}r-Rao bound on the mean squared error of the GAE, $\mu(\mathbf{x}) := \min_{i \in \mathcal{I}} |x_i|$ and $\alpha$ and $\beta$ be  fixed numbers. If
\begin{itemize}
\item $L \mu^4(\mathbf{x})/ \log L \rightarrow \infty$ as $M \rightarrow \infty$,
\item $\|\mathbf{x}\|^2$ grows slower than $L^\kappa$, for some positive
constant $\kappa$,
\item $N > C L$, where $C= \max \{C_1,C_2\}$ with $C_1:=(18 \kappa \sigma^4+1)$
and $C_2 := (9 + 4\log(\beta -1))$,
\end{itemize}
assuming that the locations of the nonzero elements of $\mathbf{x}$ are not known, there exists an estimator (namely, joint typicality decoder) for the nonzero elements of $\mathbf{x}$ with mean squared error $e_\delta^{(M)}$, such that with high probability
\begin{equation}
\nonumber \limsup_{M} \big|e_\delta^{(M)} - e_G^{(M)}\big| = 0
\end{equation}
\end{thm}
}

\section{Proof of the Main Theorem}\label{proof}

In order to establish the Main Result, we need to specify the Cram\'{e}r-Rao
bound of the GAE and define the joint typicality decoder. The following lemma gives the Cram\'{e}r-Rao
bound:

\begin{lem}\label{lemma1}
For any unbiased estimator $\mathbf{\hat{x}}$ of $\mathbf{x}$,
\begin{equation}
\mathbb{E} \big \{ \|\mathbf{\hat{x}} - \mathbf{x} \|_2^2 \big \} \ge e_G^{(M)} := \sigma^2 \operatorname{Tr} \big ((\mathbf{A}_\mathcal{I}^* \mathbf{A}_\mathcal{I})^{-1}\big )
\end{equation}
where $\mathbf{A}_\mathcal{I}$ is the sub-matrix of $\mathbf{A}$ with columns corresponding to the index set $\mathcal{I}$.
\end{lem}

\begin{proof}
Assuming that a Genie provides us with $\mathcal{I}$, we have
\begin{equation}
p_{\mathbf{y}|\mathbf{x}}(\mathbf{y};\mathbf{x}) = \frac{1}{(2 \pi)^{N/2} \sigma^N} \exp{\bigg(-\frac{1}{2\sigma^2}\|\mathbf{y} - \mathbf{A}_\mathcal{I} \mathbf{x}_\mathcal{I} \|_2^2\bigg)}.
\end{equation}
where $\mathbf{x}_\mathcal{I}$ is the sub-vector of $\mathbf{x}$ with elements corresponding to the index set $\mathcal{I}$. The Fisher information matrix is then given by
\begin{equation}
\mathbf{J}_{ij} := -\mathbb{E} \bigg \{ \frac{\partial^2 \ln p_{\mathbf{y}|\mathbf{x}}(\mathbf{y};\mathbf{x})}{\partial x_i \partial x_j} \bigg \} = \frac{1}{\sigma^2} (\mathbf{A}_\mathcal{I}^* \mathbf{A}_\mathcal{I})_{ij}
\end{equation}

Therefore, for any unbiased estimator $\mathbf{\hat{x}}$ by the Cram\'{e}r-Rao bound \cite{vantrees},
\begin{equation}\label{crb}
\mathbb{E} \big \{ \|\mathbf{\hat{x}} - \mathbf{x} \|_2^2 \big \} \ge \operatorname{Tr}(\mathbf{J}^{-1}) = \sigma^2 \operatorname{Tr}\big ((\mathbf{A}_\mathcal{I}^* \mathbf{A}_\mathcal{I})^{-1}\big )
\end{equation}
\end{proof}

Next, we state a lemma regarding the rank of sub-matrices of random i.i.d.
Gaussian matrices:
\begin{lem}\label{tek}
Let $\textbf{A}$ be a measurement matrix whose elements are i.i.d. Gaussian
 $\mathcal{N}(0,1)$, $\mathcal{J} \subset \{1,2,\cdots,M\}$ such that $|\mathcal{J}| = L$ and $\mathbf{A}_{\mathcal{J}}$ be the sub-matrix of $\mathbf{A}$ with columns corresponding to the index set $\mathcal{J}$. Then, $\operatorname{rank}(\mathbf{A}_{\mathcal{J}}) = L$ with probability 1.
\end{lem}

We can now define the notion of joint typicality. We adopt the definition from \cite{mehmet}:
\begin{defn}
We say an $N \times 1$ noisy observation vector, $\mathbf{y} = \mathbf{A} \mathbf{x} + \mathbf{n}$ and a set of indices $\mathcal{J} \subset \{1,2,\cdots,M\}$, with $|\mathcal{J}|=L$, are $\delta$-jointly typical if $\operatorname{rank}(\mathbf{A}_{\mathcal{J}})=L$ and
\begin{equation}\label{typ}
\bigg| \frac{1}{N} \|\Pi_{\mathbf{A}_\mathcal{J}}^\perp \mathbf{y}\|^2 - \frac{N-L}{N} \sigma^2 \bigg| < \delta
\end{equation}
where $\mathbf{A}_{\mathcal{J}}$ is the sub-matrix of $\mathbf{A}$ with columns corresponding to the index set $\mathcal{J}$ and $\Pi_{\mathbf{A}_\mathcal{J}}^\perp := \mathbf{I} - \mathbf{A}_\mathcal{J} (\mathbf{A}_\mathcal{J}^* \mathbf{A}_\mathcal{J} )^{-1} \mathbf{A}_\mathcal{J}^*$. We denote the event
of $\delta$-jointly typicality of $\mathcal{J}$ with $\mathbf{y}$ by $E_\mathcal{J}$.
\end{defn}
Note that we can make the assumption of $\operatorname{rank}(\mathbf{A}_{\mathcal{J}})=L$ without loss of generality, based on Lemma \ref{tek}.

\begin{defn}[Joint Typicality Decoder]
The joint typicality decoder finds a set of indices $\hat{\mathcal{I}}$ which is $\delta$-jointly typical with $\mathbf{y}$, by projecting $\mathbf{y}$ onto all the possible $M \choose L$ $L$-dimensional subspaces spanned by the columns of $\mathbf{A}$ and choosing the one satisfying Eq. (\ref{typ}).
It then produces the estimate $\mathbf{\hat{x}}^{(\hat{\mathcal{I}})}$ by projecting $\mathbf{y}$ onto the subspace spanned by $\mathbf{A}_{\hat{\mathcal{I}}}$:
\begin{equation}
\mathbf{\hat{x}}^{(\hat{\mathcal{I}})} = \begin{cases} \big(\mathbf{A}_{\hat{\mathcal{I}}}^* \mathbf{A}_{\hat{\mathcal{I}}}\big)^{-1} \mathbf{A}_{\hat{\mathcal{I}}}^* \mathbf{y}\Big|_{\hat{\mathcal{I}}}& \mbox{on } \hat{\mathcal{I},}\\
\mathbf{0} & \mbox{elsewhere}. \\
\end{cases}
\end{equation}
If the estimator does not find any set $\delta$-typical to $\mathbf{y}$, it will output the zero vector as the estimate. We denote this event by $E_0$.
\end{defn}

In what follows, we show that the joint typicality decoder has the property stated in the Main Theorem.
The next lemma from \cite{mehmet} gives non-asymptotic bounds on the probability
of the error events $E^c_{\mathcal{I}}$ and $E_{\mathcal{J}}$, averaged over
the ensemble of random i.i.d. Gaussian measurement matrices:

\begin{lem}[Lemma 3.3 of \cite{mehmet}]\label{lemma3}
For any $\delta > 0$,
\begin{equation}
\mathbb{P}\bigg( \Big| \frac{1}{N} \|\Pi_{\mathbf{A}_\mathcal{I}}^\perp \mathbf{y}\|^2 - \frac{N-L}{N} \sigma^2 \Big| > \delta \bigg) \le 2 \exp \bigg( - \frac{\delta^2}{4 \sigma^4} \frac{N^2}{N-L+\frac{2\delta}{\sigma^2} N} \bigg)
\end{equation}
and
\begin{equation}
\mathbb{P}\bigg( \Big| \frac{1}{N} \|\Pi_{\mathbf{A}_\mathcal{J}}^\perp \mathbf{y}\|^2 - \frac{N-L}{N} \sigma^2 \Big| < \delta \bigg) \le \exp \Bigg( - \frac{N-L}{4} \bigg ( \frac{\sum_{k \in \mathcal{I} \backslash \mathcal{J}} |x_k|^2 - \delta'}{\sum_{k \in \mathcal{I} \backslash \mathcal{J}} |x_k|^2 + \sigma^2} \bigg )^2 \Bigg)
\end{equation}
where $\mathcal{J}$ is an index set such that $|\mathcal{J}| = L$, $|\mathcal{I} \cap \mathcal{J}|=K < L$, rank$(\mathbf{A}_\mathcal{J})=L$ and
\begin{equation}
\nonumber \delta' := \delta \frac{N}{N-L} < \min_k |x_k|^2.
\end{equation}
\end{lem}

\begin{proof}
The proof is given in \cite{mehmet}.\\
\end{proof}

\begin{proof}[{Proof of the Main Theorem}]
Let $\mathcal{A}$ be the set of all $N \times M$ matrices with i.i.d. Gaussian
$\mathcal{N}(0,1)$ entries. Consider $\mathbf{A} \in \mathcal{A}$. Then, it is known that \cite{candes1}
for all $\mathcal{K} \in \{1,2,\cdots,M\}$ such that $|\mathcal{K}|=L$,
\begin{equation}
\mathbb{P} \bigg ( \lambda_{\max}\Big(\frac{1}{N} \mathbf{A}^*_{\mathcal{K}} \mathbf{A}_{\mathcal{K}}\Big ) > (1 + \sqrt{\alpha} + \epsilon )^2 \bigg)
\le \exp\bigg(-\frac{1}{2} M \frac{\sqrt{H(1/\beta)}}{\sqrt{\alpha \beta}} \epsilon \bigg),
\end{equation}
and
\begin{equation}
\mathbb{P} \bigg ( \lambda_{\min}\Big(\frac{1}{N} \mathbf{A}^*_{\mathcal{K}} \mathbf{A}_{\mathcal{K}}\Big ) < (1 - \sqrt{\alpha} - \epsilon )^2 \bigg)
\le \exp\bigg(-\frac{1}{2} M \frac{\sqrt{H(1/\beta)}}{\sqrt{\alpha \beta}} \epsilon \bigg).
\end{equation}
Let $\mathcal{A}^{(\alpha)}_\epsilon \subset \mathcal{A}$, such that for all $\mathbf{A} \in \mathcal{A}^{(\alpha)}_\epsilon$, the eigenvalues of $\frac{1}{N} \mathbf{A}_\mathcal{K}^*
\mathbf{A}_{\mathcal{K}}$ lie in the interval $[(1-\sqrt{ \alpha}-\epsilon)^2,(1+\sqrt{
\alpha} + \epsilon)^2]$. Clearly, we have
\[
P_\mathcal{A} \Big ( \mathcal{A}_{\epsilon}^{(\alpha)}\Big ) \ge 1 - 2 \exp\bigg(-\frac{1}{2} M \frac{\sqrt{H(1/\beta)}}{\sqrt{\alpha \beta}} \epsilon \bigg).
\]
where $P_{\mathcal{A}}(\cdot)$ denotes the Gaussian measure defined over
$\mathcal{A}$.
It is easy to show that
\begin{equation}\label{mean_ineq}
\mathbb{E}_{\mathbf{A} \in \mathcal{A}^{(\alpha)}_\epsilon} \big \{ f(\mathbf{A}) \big \} \le \frac{\mathbb{E}_{\mathbf{A} \in \mathcal{A}}\big \{ f(\mathbf{A}) \big \}}{{P_{\mathcal{A}}}\Big ( \mathcal{A}^{(\alpha)}_\epsilon\Big )},
\end{equation}
where $f(\mathbf{A})$ is any function of $\mathbf{A}$ for which $\mathbb{E}_{\mathbf{A}
\in \mathcal{A}}
\big \{ f(\mathbf{A})\big\}$ exists.
In what follows, we will upper-bound the MSE of the joint typicality decoder, averaged over all Gaussian measurement matrices in $\mathcal{A}^{(2\alpha)}_{\epsilon}$.
Let $e_\delta^{(M)}$ denote the MSE of the joint typicality decoder.
We have:
\begin{align}\label{upb}
\nonumber e_\delta^{(M)} &= \mathbb{E}_{\mathbf{n}} \Big \{ \| \mathbf{\hat{x}} - \mathbf{x}\|_2^2 \Big \}\\
\nonumber & \le \mathbb{E}_{\mathbf{n}} \Big \{ \big \| \mathbf{\hat{x}}^{({\mathcal{I}})} - \mathbf{x}\big \|_{2}^2 \Big \} +  \|\mathbf{x}\|_2^2 \mathbb{E}_{\mathbf{n}} \big \{\mathbb{I}(E_{\mathcal{I}}^c)
\big \}\\
& \quad + \sum_{\mathcal{J} \neq \mathcal{I}} \mathbb{E}_{\mathbf{n}} \Big \{ \big \| \mathbf{\hat{x}}^{({\mathcal{J}})} - \mathbf{x}\big \|_{2}^2 \mathbb{I}(E_\mathcal{J} )\Big \} := \check{e}_\delta^{(M)}
\end{align}
where $\mathbb{I}(\cdot)$ is the indicator function, $\mathbb{E}_\mathbf{n}\{\cdot \}$ denotes the expectation operator defined over the noise density, and the inequality follows from the union bound and the facts that $\mathbb{E}_{\mathbf{n}}
\{ \mathbb{I}(E_0)\} =:\mathbb{P}(E_0) \le \mathbb{P}(E_{\mathcal{I}}^c)
:= \mathbb{E}_{\mathbf{n}}
\{ \mathbb{I}(E_{\mathcal{I}}^c)\}$ and $\mathbb{I}(E_{\mathcal{I}})\le
1$. Now, averaging the upper bound $\check{e}_{\delta}^{(M)}$ (defined in
(\ref{upb}))
over $\mathcal{A}_{\epsilon}^{(2\alpha)}$ yields:
\begin{eqnarray}\label{var1}
\nonumber \mathbb{E}_\mathbf{A} \Big \{\check{e}_\delta^{(M)} \Big \}&=& \int_{{\mathcal{A}^{(2\alpha)}_{\epsilon}}} \mathbb{E}_{\mathbf{n} | \mathbf{A}} \Big \{ \big \| \mathbf{\hat{x}}^{({\mathcal{I}})} - \mathbf{x}\big \|_2^2   \Big \}dP(\mathbf{A})\\
\nonumber &+& \int_{{\mathcal{A}^{(2\alpha)}_{\epsilon}}} \|\mathbf{x}\|_2^2 \mathbb{E}_{\mathbf{n}} \big \{ \mathbb{I}(E_{\mathcal{I}}^c) \big \}dP(\mathbf{A})\\
&+& \int_{{\mathcal{A}^{(2\alpha)}_{\epsilon}}} \sum_{\mathcal{J}\neq\mathcal{I}} \mathbb{E}_{\mathbf{n} | \mathbf{A}} \Big \{ \big \| \mathbf{\hat{x}}^{({\mathcal{J}})} - \mathbf{x}\big \|_2^2  \mathbb{I}(E_{\mathcal{J}}) \Big \}dP(\mathbf{A})
\end{eqnarray}
where $dP(\mathbf{A})$ and $\mathbb{E}_\mathbf{A}$ denote the conditional
Gaussian probability measure and the expectation operator defined on the set $\mathcal{A}_{\epsilon}^{(2\alpha)}$, respectively.

The first term on the right hand side of Eq. (\ref{var1}) can be written
as:
\begin{align}
\nonumber \int_{{\mathcal{A}^{(2\alpha)}_{\epsilon}}} \mathbb{E}_{\mathbf{n} | \mathbf{A}} \bigg \{ \Big \| (\mathbf{A}_{\mathcal{I}}^* \mathbf{A}_{\mathcal{I}})^{-1} \mathbf{A}_{\mathcal{I}}^* \mathbf{y}\big|_{\mathcal{I}} - \mathbf{x}\Big \|_2^2  \bigg \} dP(\mathbf{A}) = & \mathbb{E}_{\mathbf{n},\mathbf{A}} \Big \{ \big \| (\mathbf{A}_{\mathcal{I}}^* \mathbf{A}_{\mathcal{I}})^{-1} \mathbf{A}_{\mathcal{I}}^* \mathbf{n} \big \|_2^2 \Big \}\\
=& \mathbb{E}_{\mathbf{A}} \Big \{ \sigma^2 \operatorname{Tr}(\mathbf{A}_{\mathcal{I}}^* \mathbf{A}_{\mathcal{I}})^{-1} \Big \}.
\end{align}

The second term on the right hand side of Eq. (\ref{var1}) can be upper-bounded as
\begin{align}\label{var12}
\int_{{\mathcal{A}^{(2\alpha)}_{\epsilon}}} \|\mathbf{x}\|_2^2 \mathbb{E}_\mathbf{n}\big \{\mathbb{I}(E_{\mathcal{I}}^c)
\big \} dP(\mathbf{A}) \le & \quad \frac{2 \|\mathbf{x}\|_2^2}{{P_{\mathcal{A}}\Big(
\mathcal{A}_\epsilon^{(2\alpha)}\Big)}} \exp \bigg( - \frac{\delta^2}{4 \sigma^4} \frac{N^2}{N-L+\frac{2\delta}{\sigma^2} N} \bigg)
\end{align}
where the inequality follows from Lemma \ref{lemma3} and Eq. (\ref{mean_ineq}).

Now, let $\delta' = \delta N/(N-L) = \zeta \mu^2(\mathbf{x})$ for some $2/3 < \zeta < 1$. Recall that by hypothesis,
$\mu^2(\mathbf{x})$ goes to zero slower than $\sqrt{\frac{\log L}{L}}$. Suppose that $L$ is
large enough, so that $\sigma^2 \ge 2 \zeta \mu^2(\mathbf{x})$. Then, it
can be shown that the right hand side of Eq. (\ref{var12}) decays faster than $1 / L^c$, where $c:= \zeta^2 C_0/2 \sigma^4$. Now, by the hypothesis
of $N > (18 \kappa \sigma^4 + 1)L$
we have $C_0 > 9\kappa \sigma^4/2$, therefore $c > \kappa$, and hence the right hand side of Eq. (\ref{var12}) goes to zero faster than $1/L^{(c-\kappa)}$, as $M,N \rightarrow \infty$.

\textbf{Remark}. If we assume the weaker condition that $\mu(\mathbf{x})
= \mu_0$, constant independent of $M, N$ and $L$, then $\delta$ will be constant
and the error exponent in Eq. (\ref{var12}) decays to zero exponentially
fast. Hence, as long as $\|\mathbf{x}\|_2^2$ grows polynomially, the whole
term on the right hand side of Eq. (\ref{var12}) decays to zero exponentially fast. Therefore, the claim of the theorem
holds with overwhelming probability, \emph{i.e.}, probability of failure exponentially
small in $L$ (rather than with high probability, which refers to the case
of failure
probability polynomially small in $L$).

Finally, consider the term corresponding to $\mathcal{J}$ in the third expression on the right hand side of Eq. (\ref{var1}). This term can be simplified as
\begin{align}\label{var44}
\nonumber  \mathbb{E}_{\mathbf{n}, \mathbf{A}} \Big \{ \big \|\mathbf{\hat{x}}^{({\mathcal{J}})} - \mathbf{x}\big \|_2^2  \mathbb{I}(E_\mathcal{J} )\Big \} &= \mathbb{E}_{\mathbf{n},\mathbf{A}}
\Big \{ \big \|\mathbf{\hat{x}}^{({\mathcal{J}})}_{\mathcal{J}} - \mathbf{x}_{\mathcal{J}}\big \|_2^2 \mathbb{I}(E_\mathcal{J} )\Big \} + \mathbb{E}_{\mathbf{n},\mathbf{A}}
\Big \{ \big \| \mathbf{x}_{\mathcal{I}\backslash \mathcal{J}}\big \|_2^2 \mathbb{I}(E_\mathcal{J} )\Big \}\\
\nonumber & = \underbrace{\mathbb{E}_{\mathbf{n},\mathbf{A}} \Big \{ \big \| (\mathbf{A}_{\mathcal{J}}^* \mathbf{A}_{\mathcal{J}})^{-1} \mathbf{A}_{\mathcal{J}}^* \mathbf{A}
\mathbf{x} - \mathbf{x}_{\mathcal{J}}\big\|_2^2 \mathbb{I}(E_\mathcal{J} ) \Big \}}_{T_1}\\
\nonumber &+ \underbrace{\mathbb{E}_{\mathbf{n}, \mathbf{A}} \Big
\{\big\|(\mathbf{A}_{\mathcal{J}}^* \mathbf{A}_{\mathcal{J}})^{-1} \mathbf{A}_{\mathcal{J}}^* \mathbf{n} \big \|_2^2\mathbb{I}(E_\mathcal{J} ) \Big \}}_{T_2}\\
& + \mathbb{E}_{\mathbf{n},\mathbf{A}}
\Big \{ \big \| \mathbf{x}_{\mathcal{I}\backslash \mathcal{J}}\big \|_2^2 \mathbb{I}(E_\mathcal{J} )\Big \}.
\end{align}
where the first equality follows from the fact that $\hat{\mathbf{x}}^{(\mathcal{J})}$
is zero outside of the index set $\mathcal{J}$, and the second equality
follows from the assumption that $\mathbf{n}$ is zero-mean and independent of $\mathbf{A}$ and $\mathbf{x}$. Now,  the term $T_1$ can be further simplified
as
\begin{align}
\nonumber T_1 &=\mathbb{E}_{\mathbf{n},\mathbf{A}} \Big \{ \big \| (\mathbf{A}_{\mathcal{J}}^* \mathbf{A}_{\mathcal{J}})^{-1} \mathbf{A}_{\mathcal{J}}^* (\mathbf{A}_{\mathcal{J}}
\mathbf{x}_{\mathcal{J}}+ \mathbf{A}_{\mathcal{I} \backslash
\mathcal{J}}\mathbf{x}_{\mathcal{I} \backslash \mathcal{J}}) - \mathbf{x}_\mathcal{J}
\big\|^2_2\mathbb{I}(E_\mathcal{J} )\Big \}\\
\nonumber & = \mathbb{E}_{\mathbf{n},\mathbf{A}} \Big \{ \big \| (\mathbf{A}_{\mathcal{J}}^* \mathbf{A}_{\mathcal{J}})^{-1} \mathbf{A}_{\mathcal{J}}^*  \mathbf{A}_{\mathcal{I} \backslash \mathcal{J}}\mathbf{x}_{\mathcal{I} \backslash \mathcal{J}}
\big\|^2_2\mathbb{I}(E_\mathcal{J} )\Big \}
\end{align}
Invoking the sub-multiplicative property of matrix norms and applying the Cauchy-Shwarz inequality to the mean of the product of non-negative random variables, we can further upper bound $T_1$ as:
\begin{align}
\nonumber T_1  &\le \mathbb{E}_{\mathbf{A}} \Big \{ \big \| (\mathbf{A}_{\mathcal{J}}^* \mathbf{A}_{\mathcal{J}})^{-1}\big \|_2^2 \big\| \mathbf{A}_{\mathcal{J}}^*  \mathbf{A}_{\mathcal{I} \backslash \mathcal{J}}\big \|_2^2 \big \|\mathbf{x}_{\mathcal{I} \backslash \mathcal{J}}
\big\|^2_2\mathbb{I}(E_\mathcal{J} )\Big \}\\
& \le \mathbb{E}_{\mathbf{A}} \Big \{ \big \| (\mathbf{A}_{\mathcal{J}}^* \mathbf{A}_{\mathcal{J}})^{-1}\big\|_2^2\Big \} \mathbb{E}_{\mathbf{A}} \Big \{\big \| \mathbf{A}_{\mathcal{J}}^* \mathbf{A}_{\mathcal{I}\backslash \mathcal{J}} \big \|_2^2
\Big \} \big\|\mathbf{x}_{\mathcal{I} \backslash \mathcal{J}}\big\|_2^2
 \mathbb{E}_{\mathbf{n},\mathbf{A}} \Big \{\mathbb{I}(E_\mathcal{J} ) \Big \}
\end{align}

For simplicity, let $\epsilon = \sqrt{2\alpha}$. Now, consider the  term $\mathbb{E}_{\mathbf{A}} \Big \{ \big \| (\mathbf{A}_{\mathcal{J}}^* \mathbf{A}_{\mathcal{J}})^{-1}\big\|_2^2\Big \}$, the average of the squared spectral norm of the matrix
$(\mathbf{A}_{\mathcal{J}}^* \mathbf{A}_{\mathcal{J}})^{-1}$. Since $\mathbf{A}
\in \mathcal{A}_\epsilon$, the minimum eigenvalue of $\frac{1}{N}\mathbf{A}_{\mathcal{J}}^* \mathbf{A}_{\mathcal{J}}$ is greater than or equal to $(1-2\sqrt{2\alpha})^2$, since $|\mathcal{J}| = L < 2L$, and $\mathbf{A}^*_\mathcal{J}
\mathbf{A}_\mathcal{J}$ is a sub-matrix of $\mathbf{A}^*_\mathcal{K}
\mathbf{A}_\mathcal{K}$ for some $\mathcal{K}$ such that $\mathcal{J} \subset
\mathcal{K}$ and $|\mathcal{K}| = 2L$ \cite[Theorem 4.3.15]{matrix}. Similarly, for the second term, $\mathbf{A}_{\mathcal{J}}^* \mathbf{A}_{\mathcal{I}\backslash \mathcal{J}}$ is a sub-matrix of $\mathbf{A}_{\mathcal{J} \cup
\mathcal{I}}^* \mathbf{A}_{\mathcal{J}\cup \mathcal{I}} - \mathbf{I}$. Since $|\mathcal{J}
\cup \mathcal{I}| \le 2L$,
the spectral norm of $\frac{1}{N}\mathbf{A}_{\mathcal{J}}^* \mathbf{A}_{\mathcal{I}\backslash
\mathcal{J}}$
is upper bounded by $(1+2\sqrt{2\alpha})^2 - 1 = 8 \alpha + 4 \sqrt{2 \alpha}$
\cite[Theorem 4.3.15]{matrix}, \cite{NT09}.
Finally, since the averaging is over $\mathcal{A}^{(2\alpha)}_\epsilon$,
the same bounds hold for the average of the matrix norms.
Hence, $T_1$ can be upper bounded by
\begin{align}
\nonumber T_1 \le \frac{(8 \alpha +4\sqrt{2\alpha})^2}{(1-2\sqrt{2\alpha})^4}
 \mathbb{E}_{\mathbf{n},\mathbf{A}} \Big \{\mathbb{I}(E_\mathcal{J} ) \Big \} \|
\mathbf{x}\|_2^2.
\end{align}

Note that by the hypothesis of $N > C_2 L$, we have $\alpha < 1/9$. Hence,
$2 \sqrt{2\alpha} = \frac{2 \sqrt{2}}{3} < 1$. Also, $T_2$ is similarly upper bounded by $\sigma^2 \mathbb{E}_\mathbf{A} \big \{ \mbox{Tr}(\mathbf{A}_{\mathcal{J}}^* \mathbf{A}_{\mathcal{J}})^{-1} \big \}\mathbb{E}_{\mathbf{n},\mathbf{A}} \big\{ \mathbb{I} (E_\mathcal{J})
\big \}$, which can be further bounded by
\[
T_2 \le \frac{\alpha \sigma^2}{(1-2\sqrt{2\alpha})^2}\mathbb{E}_{\mathbf{n},\mathbf{A}} \Big \{\mathbb{I}(E_\mathcal{J} ) \Big \}.
\]
Therefore, Eq. (\ref{var44}) can be further bounded as
\begin{align}
\nonumber \mathbb{E}_{\mathbf{n}, \mathbf{A}} \Big \{ \big \|\mathbf{\hat{x}}^{({\mathcal{J}})} - \mathbf{x}\big \|_2^2  \mathbb{I}(E_\mathcal{J} )\Big \} & \le \underbrace{\bigg [ \Big(1+{\textstyle
\frac{(8 \alpha +4\sqrt{2\alpha})^2}{(1-2\sqrt{2\alpha})^4}}\Big)\|
\mathbf{x}\|_2^2 + {\textstyle \frac{\alpha \sigma^2}{(1-2\sqrt{2\alpha})^2}}\bigg]}_{\mathscr{C}(\mathbf{x})}\mathbb{E}_{\mathbf{n},\mathbf{A}} \Big \{\mathbb{I}(E_\mathcal{J} ) \Big \}
\end{align}

Now, from Lemma \ref{lemma3} and Eq. (\ref{mean_ineq}), we have
\begin{align}\label{var4}
\nonumber & \sum_{\mathcal{J} \neq \mathcal{I}}\mathbb{E}_{\mathbf{n}, \mathbf{A}} \Big \{ \big \|\mathbf{\hat{x}}^{({\mathcal{J}})} - \mathbf{x}\big \|_2^2  \mathbb{I}(E_\mathcal{J} )\Big \}\\
\nonumber & \qquad \qquad \qquad \qquad \le {\mathscr{C}(\mathbf{x})} \sum_{\mathcal{J} \neq \mathcal{I}}\mathbb{E}_{\mathbf{n},\mathbf{A}} \Big \{\mathbb{I}(E_\mathcal{J} ) \Big \}\\
& \qquad \qquad \qquad \qquad\le \frac{\mathscr{C}(\mathbf{x})}{P_{\mathcal{A}}\big(\mathcal{A}^{(2\alpha)}_{\epsilon}\big)}\sum_{\mathcal{J}\neq\mathcal{I}} \exp \Bigg( - \frac{N-L}{4} \bigg ( \frac{\sum_{k \in \mathcal{I} \backslash \mathcal{J}} |x_k|^2 - \delta'}{\sum_{k \in \mathcal{I} \backslash \mathcal{J}} |x_k|^2 + \sigma^2} \bigg )^2 \Bigg),
\end{align}
where $x_k$ denotes the $k$th component of $\mathbf{x}$.
The number of index sets $\mathcal{J}$ such that $|\mathcal{J} \cap \mathcal{I}| = K < L$ is upper-bounded by ${L \choose L-K}{M-L \choose L-K}$. Also, $\sum_{k \in \mathcal{I} \backslash \mathcal{J}} |x_k|^2 \ge (L-K) \mu^2(\mathbf{x})$. Therefore, we can rewrite the right hand side of Eq. (\ref{var4}) as
\begin{align}\label{var2}
\nonumber &\frac{\mathscr{C}(\mathbf{x})}{P_{\mathcal{A}}\big(\mathcal{A}^{(2\alpha)}_{\epsilon}\big)} \sum_{K=0}^{L-1} {L \choose L-K}{M-L \choose L-K} \exp \Bigg( - \frac{N-L}{4} \bigg ( \frac{(L-K)\mu^2(\mathbf{x}) - \delta'}{(L-K)\mu^2(\mathbf{x}) + \sigma^2} \bigg )^2 \Bigg)\\
&=\frac{\mathscr{C}(\mathbf{x})}{P_{\mathcal{A}}\big(\mathcal{A}^{(2\alpha)}_{\epsilon}\big)} \sum_{K'=1}^{L} {L \choose K'}{M-L \choose K'} \exp \Bigg( - \frac{N-L}{4} \bigg ( \frac{K'\mu^2(\mathbf{x}) - \delta'}{K'\mu^2(\mathbf{x}) + \sigma^2} \bigg )^2 \Bigg)
\end{align}
We use the inequality
\begin{equation}
{L \choose K'} \le \exp \Big ( K' \log \Big ( \frac{L e}{K'} \Big ) \Big )
\end{equation}
in order to upper-bound the $K'$th term of the summation in Eq. (\ref{var2}) by
\begin{eqnarray}
\exp \bigg ( L \frac{K'}{L} \log \Big ( \frac{e}{\frac{K'}{L}} \Big ) + L \frac{K'}{L} \log \frac{(\beta-1)e}{\frac{K'}{L}} - C_0L \Big ( \frac{L \frac{K'}{L} \mu^2(\mathbf{x}) - \delta'}{L \frac{K'}{L} \mu^2(\mathbf{x}) + \sigma^2} \Big )^2 \bigg )
\end{eqnarray}
where $C_0 := (N-L)/4L$. We define

\begin{equation}
f(z) := L z \log \frac{e}{z} + L z \log \frac{(\beta-1)e}{z} - C_0 L \Big ( \frac{L z \mu^2(\mathbf{x}) - \delta'}{L z \mu^2(\mathbf{x}) + \sigma^2} \Big )^2
\end{equation}

By Lemmas 3.4, 3.5 and 3.6 of \cite{mehmet}, it can be shown that $f(z)$ asymptotically attains its maximum at either $z=\frac{1}{L}$ or $z=1$, if $L \mu^4(\mathbf{x})/ \log L \rightarrow \infty$ as $N \rightarrow \infty$. Thus, we can upper-bound the right hand side of Eq. (\ref{var4}) by
\begin{eqnarray}\label{var3}
\frac{\mathscr{C}(\mathbf{x})}{P_{\mathcal{A}}\big(\mathcal{A}^{(2\alpha)}_{\epsilon}\big)}\sum_{K'=1}^{L} \exp \Big ( \max \{f({1}/{L}),f(1) \} \Big ).
\end{eqnarray}
The values of $f(1/L)$ and $f(1)$ can be written as:
\begin{equation}\label{f1l}
f(1/L) = 2 \log L + 2 + \log(\beta -1) - C_0L \Big ( \frac{ \mu^2(\mathbf{x}) - \delta'}{\mu^2(\mathbf{x}) + \sigma^2} \Big )^2
\end{equation}
and
\begin{equation}\label{f1}
f(1)=L(2 + \log(\beta -1)) - C_0 L \Big ( \frac{ L \mu^2(\mathbf{x}) - \delta'}{L \mu^2(\mathbf{x}) + \sigma^2} \Big )^2.
\end{equation}
Since $C_0 > 2 + \log(\beta-1)$ due to the assumption of $\alpha < 1/(9 + 4\log(\beta-1))$, both $f(1)$ and $f(1/L)$ will grow to $-\infty$ linearly as $N \rightarrow \infty$. Hence, the exponent in Eq. (\ref{var3}) will grow to $-\infty$ as $N \rightarrow \infty$, since $\| \mathbf{x} \|_2^2$ in
$\mathscr{C}(\mathbf{x})$ grows polynomially in $L$ and $P_{\mathcal{A}}\big(\mathcal{A}^{(2\alpha)}_{\epsilon}\big)$
tends to $1$ exponentially fast. Let $L$ be large enough such that
\[
\frac{\mathscr{C}(\mathbf{x})}{P_{\mathcal{A}}\big(\mathcal{A}^{(2\alpha)}_{\epsilon}\big)} \sum_{K'=1}^{L} \exp \Big ( \max \{f({1}/{L}),f(1) \} \Big ) \le \frac{1}{L^{c
- \kappa}},
\]
where $c = \zeta^2 C_0 / 2 \sigma^4$. Now, by Markov's inequality we have:
\begin{equation}
P\Bigg(\bigg|\check{e}_\delta^{(M)} - {\sigma^2  \operatorname{Tr}(\mathbf{A}_{\mathcal{I}}^* \mathbf{A}_{\mathcal{I}})^{-1}} \bigg| > \frac{1}{L^{\frac{c-\kappa}{2}}} \Bigg ) \le \frac{2}{L^{\frac{c-\kappa}{2}}}.
\end{equation}

Hence, for any measurement matrix $\mathbf{A}$ from the restricted ensemble of i.i.d. Gaussian matrices $\mathcal{A}_{\epsilon}^{(2\alpha)}$, we have
\begin{equation}\label{stat}
\bigg|\check{e}_\delta^{(M)} - {\sigma^2 \operatorname{Tr}(\mathbf{A}_{\mathcal{I}}^* \mathbf{A}_{\mathcal{I}})^{-1}} \bigg|  \le \frac{1}{L^{\frac{c-\kappa}{2}}}
\end{equation}
with probability exceeding
\begin{equation}
1 - \frac{2}{L^{\frac{c-\kappa}{2}}}.
\end{equation}
Now, since the probability measure is defined over $\mathcal{A}_{\epsilon}^{(2\alpha)}$,
the statement of Eq. (\ref{stat}) holds with probability exceeding
\begin{equation}
1 - 2 \exp\bigg(-\frac{1}{2} M \frac{\sqrt{2H(2/\beta)}}{\sqrt{\beta}}  \bigg) - \frac{2}{L^{\frac{c-\kappa}{2}}}
\end{equation}
over all the measurement matrices in the Gaussian ensemble $\mathcal{A}$.

Recall that the expression $\sigma^2\operatorname{Tr}(\mathbf{A}^*_\mathcal{I} \mathbf{A}_\mathcal{I})^{-1}$ is indeed the  Cram\'{e}r-Rao bound of the GAE. Noting that $\check{e}_{\delta}^{(M)}$ is an upper bound on $e_{\delta}^{(M)}$ concludes the proof the the main theorem.
\end{proof}

\bibliographystyle{unsrt}

\bibliography{crbbib}

\end{document}